\documentclass[ twocolumn, aps, preprintnumbers]{revtex4}
\usepackage{amsfonts,amssymb,amsmath,amsthm,graphicx,verbatim}
\usepackage[usenames,dvipsnames]{color}
\usepackage{epic} 
\usepackage{pict2e}   
\usepackage{calc}
\usepackage{times}
\usepackage{graphicx}
\usepackage{multirow}
\usepackage{color}
\usepackage{caption,subcaption}
\usepackage{hyperref,doi}
\newtheorem*{remark}{Remark}

\newtheorem{theorem}{Theorem}
\newtheorem{lemma}{Lemma}

\newtheorem{corollary}{Corollary}

\newcommand*{\cH}{\mathcal{H}}

\newcommand*{\cK}{\mathcal{K}}

\newcommand*{\cS}{\mathcal{S}}

\newcommand*{\id}{\mathbf{1}}

\newcommand*{\tr}{\mathrm{tr}}
\newcommand*{\ket}[1]{| #1 \rangle}

\newcommand*{\eps}{\varepsilon}

\begin{document}
\title{Confidence Polytopes in Quantum State Tomography}

\author{Jinzhao \surname{Wang}} \author{Volkher B. \surname{Scholz}} \author{Renato \surname{Renner}}
\affiliation{Institute for Theoretical Physics, ETH Zurich,
  Switzerland}

\begin{abstract}

Quantum State Tomography is the task of inferring the state of a quantum system from  measurement data. A reliable tomography scheme should not only report an estimate for that state, but also well-justified error bars. These may be specified in terms of confidence regions, i.e., subsets of the state space which contain the system's state with high probability. Here, building upon a quantum generalisation of Clopper-Pearson confidence intervals | a notion known from classical statistics | we present a simple and reliable  scheme for generating confidence regions. These have the shape of a polytope and can be computed efficiently.  We  provide several examples to demonstrate the practical  usability of the scheme in experiments.
\end{abstract}

\maketitle

\section{Introduction}

\emph{Quantum State Tomography (QST)} may be regarded as the quantum variant of statistical estimation theory. Given data obtained from measuring a quantum system, the goal is to estimate the system's state. QST has become an increasingly important tool in experimental physics, especially in the area of quantum information technology.  Accordingly, a lot of work has been put into the development of techniques to increase its efficiency. Among them are methods to reduce the number of different measurements needed  and to keep the (generally unfavourable) scaling of the amount of required data  in the dimension of the system under control~\cite{Cramer,Gross,Aaronson,Ferrie,Huszar,Toth,Goyeneche,Struchalin,Chapman}.  

Nonetheless, only relatively little attention has  been paid to the problem of \emph{statistical errors} in QST.  Statistical errors are due to unavoidable fluctuations, resulting from the fact that the  collected data always represents a finite sample.  In other words, they  are those errors that  remain even if the experiment is implemented perfectly and shielded from any environmental noise. 

In experimental sciences, statistical errors are generally reported in terms of error bars, which are obtained by standard methods from classical statistics.  In the context of quantum information, techniques to determine error estimates have been developed for specific tasks, such as  entanglement verification and quantum metrology~\cite{Walter,Sugiyama15,Baur,Rosset,Blume2010verification,Rehacek,Jungnitsch}. These are however not universal enough to be applicable to QST.  In fact,  an agreed-upon scheme for reporting the accuracy of estimates in QST does not seem to exist. Experimental results in QST are therefore often stated without error bars, or with error bars that do not have a well-defined operational meaning. A widespread approach is to use point estimators for the system's state, such as \emph{Maximum Likelihood Estimation (MLE)}~\cite{Hradil97,Hradiletal04} (for examples, see \cite{Blattetal04,Zeilingeretal05,BlattWineland08,Wallraffetal09}), and take the width of the likelihood function as a measure for their accuracy~\cite{Sugiyama11}. Another common, heuristic, method to determine the accuracy is  numerical \emph{boostrapping}, or \emph{resampling}~\cite{Home09,Efron93}. The resulting error bars then correspond to the variance of the point estimators. But since these are generally highly biased, they do not correctly reflect the uncertainty in the state estimate~(see \cite{Blume12} for a discussion). 

These problems can be avoided with methods that, rather than giving point estimates, yield \emph{regions} in state space. The idea is that these regions contain, with high probability, the (unknown) state, $\rho$, i.e., the state in which the system was prepared.  Depending on what  is meant by ``high probability'', one talks about  \emph{credibility regions} or \emph{confidence regions}.

\emph{Credibility regions} are motivated by the Bayesian approach to probability theory, where probabilities are interpreted as measures for personal belief or knowledge~\cite{Blume10,Shang,Ferrie14,Granade}. To use this approach in QST, it is necessary to specify a \emph{prior}, i.e., a probability distribution over the possible states~$\rho$, that reflects one's personal belief  before considering the measurement data.  The corresponding credibility region obtained from QST then has the property that it contains $\rho$ with high probability according to the \emph{posterior} belief, i.e., the updated belief one would have after taking into account the measurement data. The reported credibility region thus has a well-defined operational meaning | but only for those who agree with the prior. Unfortunately, there is no unique natural choice for the latter; even when demanding certain symmetries, the class of possible priors is usually infinitely large.  

 \emph{Confidence regions} avoid this prior-dependence. While they are generally larger than the credibility regions of the Bayesian approach, they contain the unknown state~$\rho$ with high probability | independently of what the prior was. Currently, there exist two approaches to obtain confidence regions. One of them, due to Blume-Kohout and Glancy \emph{et al.}, uses a construction based on the computation of likelihood ratios~\cite{Blume12, Glancy}. Although supported by heuristic arguments, it has however, to the best of our knowledge, not been established rigorously that the constructed regions are valid confidence regions. In the other approach, due to Christandl \& Renner~\cite{ChristandlRenner},  confidence regions are constructed by extending  credibility regions for a particular symmetric prior.  While the validity of these regions has been proved rigorously, their size is far from optimal (see the discussion below). 

In this paper, we propose an alternative method to determine  confidence regions in QST. It is based upon a generalisation of a notion from classical statistics, known as \emph{Clopper-Pearson confidence intervals}~\cite{ClopperPearson}. Given data from any informationally complete measurement, the corresponding confidence regions have the shape of a polytope  (see Fig.~\ref{fig:polytope}), with facets that can be  computed efficiently. As we shall demonstrate, this simple structure can also be exploited to optimise the choice of tomographic measurements for more accuracy.

\begin{figure}[t]
  \centering
  \begin{subfigure}[b]{0.49\linewidth}
    \includegraphics[width=\linewidth]{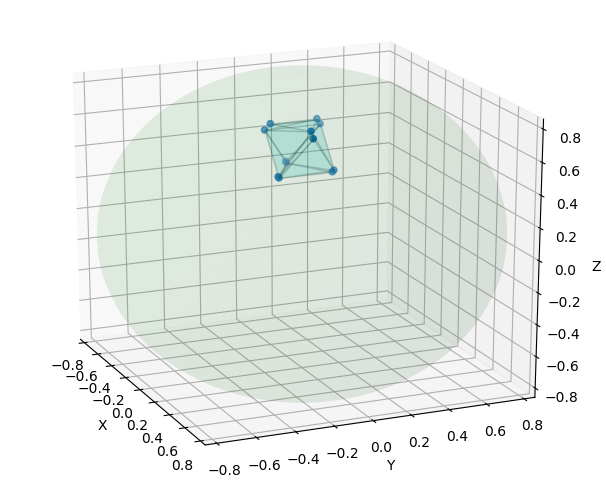}
    \caption{SIC POVM}
    \label{fig:SIC-POVM}
  \end{subfigure}
  \begin{subfigure}[b]{0.49\linewidth}
    \includegraphics[width=\linewidth]{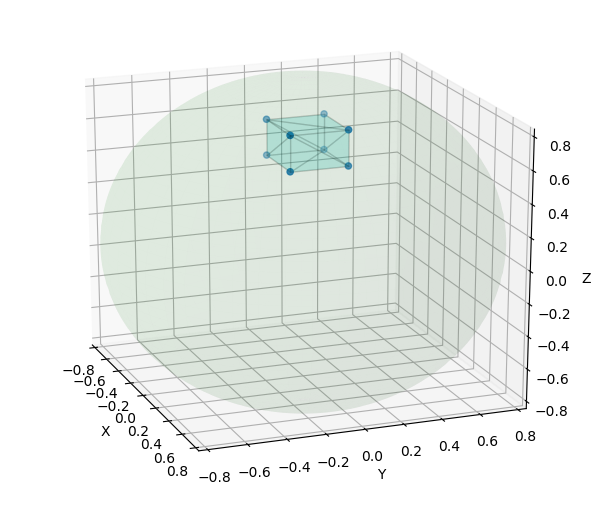}
    \caption{MUB POVM}
    \label{fig:MUB-POVM}
  \end{subfigure}
  \caption{{\bf Confidence  polytopes for QST on a qubit.} The plots show  confidence polytopes with confidence level $0.999$, obtained from data of simulated measurements on a single qubit. The polytopes lie within the Bloch sphere, which represents the entire state space of the qubit.  In (a), the measurement operators are chosen such that they form a Symmetrical Informationally Complete (SIC) POVM \cite{Renes}.   In this case, the resulting confidence polytope is the intersection of two tetrahedrons whose normal directions are given by the measurement directions. In (b), the measurement POVM is defined by a set of Mutually Unbiased Bases (MUB) \cite{Ingemar07}, and the polytope is a rectangular box whose normal directions are given by the three Pauli operators. 
 }
  \label{fig:polytope}
\end{figure}

The remainder of this paper is structured as follows. In Section~\ref{sec_construction} we describe our method to construct confidence polytopes from measurement data and formalise their operational meaning by Theorem~\ref{main}.  In Section~\ref{sec_optimise} we show how to extract, from the geometry of the polytopes,  information about the statistical uncertainty in the data, from which it is possible to determine further measurements to reduce this uncertainty.  Section~\ref{sec_comparison} is devoted to an assessment of the optimality of our confidence polytopes, which we do by comparing them to credibility regions obtained from particularly symmetric priors (Hilbert-Schmidt priors). We conclude in Section~\ref{sec_discussion} with a discussion and suggestions for future work.
 
\section{Construction of Confidence Polytopes} \label{sec_construction}

In classical estimation theory, one of the most basic problems is to determine the bias~$P$ of a biased coin from a given sample of tosses. The  Clopper-Pearson interval solves this problem ``exactly'', i.e., without involving approximations. In particular, the interval represents a reliable confidence region for~$P$, even in extreme cases, e.g., when $P\approx 0$ or $P \approx1$, in which other schemes may fail. This feature turns out to be crucial for QST, where the measurement statistics often contains such extreme cases, especially when the unknown state is close to the boundary of the state space.

In QST, one usually considers the following scenario (see~\cite{ChristandlRenner} for a more general treatment, which does not assume repeated preparations). A $d$-dimensional quantum system is repeatedly prepared in the same unknown state $\rho$, i.e., an element of the set $\cS(\cH_d)$ of density operators on a $d$-dimensional Hilbert space. After each preparation, a measurement, described by a Positive-Operator Valued Measure (POVM) on $\cH_d$ with elements $E_i$, for $i=1, \ldots, k$, is applied. After $n$ such preparation-and-measurement rounds, the data can be brought into the form of a $k$-tuple $\mathbf{n}:=(n_1,\dots, n_k)$, where $n_i$ denotes the number of occurrences of an outcome corresponding to the POVM element $E_i$. One may therefore equivalently regard $\mathbf{n}$ as the outcome associated to a POVM element of the form  $B^{\mathbf{n}}=\frac{n!}{\prod_i n_i!} \bigotimes_i E_i^{\otimes n_i}$ applied to~$\rho^{\otimes n}$. 

We are interested in constructing a QST procedure that computes, from the measurement outcome $\mathbf{n}$, a confidence region, denoted by $\Gamma(B^{\mathbf{n}})$, for any desired \emph{confidence level}~${1-\eps}$, where $\eps > 0$.  This means that, except with probability $\eps$, the unknown state $\rho$ is contained in $\Gamma(B^{\mathbf{n}})$, i.e., 
\begin{align*}
  \mathrm{Pr}[\rho \in \Gamma(B^{\mathbf{n}})] \geq {1-\eps} \ .
\end{align*}
 Crucially, this bound is supposed to hold for any arbitrary~$\rho$ (which would not be the case for a credibility region, where $\rho$ must be sampled from a given prior), whereas $\mathrm{Pr}[\cdot]$ should  be understood as the probability taken over all possible outcomes $\mathbf{n}$~\footnote{The bound does not necessarily hold conditioned on the event that a particular outcome~$\mathbf{n}$ occurred, i.e., the conditional probability that $\rho \in \Gamma(B^{\mathbf{n}})$ may be larger or smaller than $1-\eps$; see~\cite{ChristandlRenner} for a discussion.}.

We break down the problem of constructing  confidence regions into determining a confidence half-space for each individual POVM element $E_i$, depending on the measured frequency~$\frac{n_i}{n}$. The idea is that each half-space  corresponds to a one-sided Copper-Pearson confidence interval. The intersection of all such half-spaces, for all the POVM elements, then forms a confidence region, as asserted by the following theorem. For its formulation, we use the binary relative entropy, which is defined as $D(x\|y)=x\log(\frac{x}{y})+(1-x)\log(\frac{1-x}{1-y})$.

\begin{theorem} \label{main}
  Consider a QST setup as described above, with unknown state $\rho \in \cS(\cH_d)$ and measurements defined by a POVM $\{E_i\}$.   Let ${1-\eps}$ be the desired confidence level and let $\eps_i > 0$ be such that $\sum_i \eps_i = \eps$. For any possible measurement data $\mathbf{n} = (n_1, \ldots, n_k)$ and for any~$i$, define 
 \begin{align*}
   \Gamma_i(n_i) = \bigl\{\sigma \in \cS(\cH_d): \, \tr (E_i \sigma) \leq \frac{n_i}{n} + \delta(n_i, \eps_i) \bigr\}\end{align*}
    with $\delta(n_i, \eps_i)$  the positive root of $D(\frac{n_i}{n}\|\frac{n_i}{n}+\delta)=-\frac{1}{n}\mathrm{log}(\eps_i)$. Then
  \begin{align*}
  \Gamma(B^{\mathbf{n}}):=\bigcap_i \Gamma_i(n_i) 
  \end{align*}
 is a confidence region with confidence level ${1-\eps}$.
\end{theorem}

\begin{proof} See Appendix~\ref{appA}. \end{proof}

Given a family $\{\lambda_i\}$ of generalised Pauli matrices  with the orthogonality relation $\tr \lambda_i \lambda_j =2\delta_{i j}$, we can embed the space $\cS(\cH_d)$ of density operators~$\rho$ into the Euclidean space $\smash{\mathbb{R}^{d^2-1}}$ of vectors~$\mathbf{r}$ via the relation~\cite{BertlmannKrammer}
\begin{align*}
\rho=\frac{1}{d}(\id+\sqrt{\frac{d(d-1)}{2}} \mathbf{r} \cdot \boldsymbol{\lambda}) \ .
\end{align*}
The Euclidean metric on $\smash{\mathbb{R}^{d^2-1}}$ then corresponds to the Hilbert-Schmidt metric for density matrices~\cite{BEN}. Similarly, we can parametrise  each POVM element $E_i$ by a vector $\boldsymbol{\eta}_i$ in $\mathbb{R}^{d^2-1}$, i.e., 
\begin{align*}
E_i=\frac{1}{m_i}(\id+\sqrt{\frac{d(d-1)}{2}} \boldsymbol{\eta}_i \cdot \boldsymbol{\lambda}) \ ,
\end{align*}
where $m_i$ are constants that satisfy $\sum_i \frac{1}{m_i}=1$. Theorem~\ref{main} may now be rephrased in terms of these representations in $\smash{\mathbb{R}^{d^2-1}}$.

\begin{corollary} \label{mainc}
Consider a QST setup as in Theorem~\ref{main}, with an unknown state~$\rho$ parametrised by $\mathbf{r} \in \mathbb{R}^{d^2-1}$ and POVM elements $E_i$ parametrised by $\boldsymbol{\eta}_i \in \mathbb{R}^{d^2-1}$. Then the intersection of the embedding of the state space $\cS(\cH_d)$ in $\smash{\mathbb{R}^{d^2-1}}$ with the half-spaces given by 
\begin{align} \label{eq_polytopefacets}
1+(d-1)\mathbf{r} \cdot \boldsymbol{\eta}_i\leq m_i(\frac{n_i}{n}+\delta(n_i,\eps_i)) \ .
\end{align}
is a confidence region with confidence level~${1-\eps}$. 
\end{corollary}

  If the POVM is informationally complete then the confidence region is contained in a polytope in $\smash{\mathbb{R}^{d^2-1}}$, the \emph{confidence polytope}, whose facets are defined by~\eqref{eq_polytopefacets}.  Note that one may, in an obvious manner, also combine data obtained from QST measurements with respect to different POVMs on the same system. The resulting confidence region then corresponds to the intersection of the confidence polytopes belonging to each POVM, with appropriately adapted confidence levels (see Appendix~\ref{appA}). Furthermore, given measurement data for a fixed POVM, one may refine the polytope with additional facets, obtained by grouping POVM elements together to form new such elements (see Appendix~\ref{appB} for more details). Note also that the family $\{\lambda_i\}$ used for the embedding of $\cS(\cH_d)$ into  $\smash{\mathbb{R}^{d^2-1}}$  is arbitrary, up to the constraint that it satisfies the orthogonality relation above.  In particular, the $\lambda_i$'s may be chosen to match the elements of the POVM that define the tomographic measurements. The statistical errors belonging to each measurement direction can then be easily read off from the polytope (as discussed in the next section).
 
\begin{figure*}
  \centering
  \begin{subfigure}[t]{0.19\linewidth}
    \includegraphics[width=\linewidth]{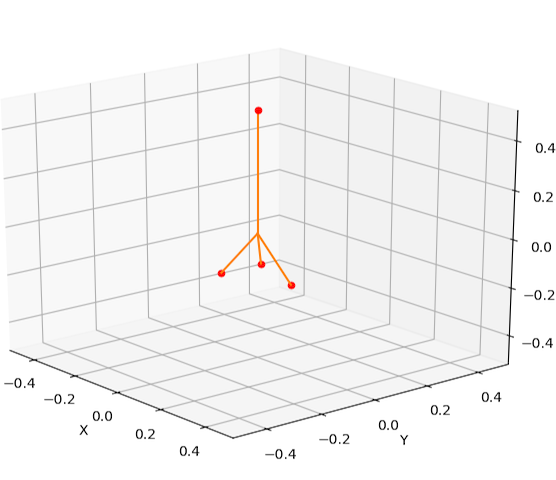}
    \caption{skewed SIC POVM}
    \label{fig:demo0}
  \end{subfigure}
  \begin{subfigure}[t]{0.34\linewidth}
    \includegraphics[width=\linewidth]{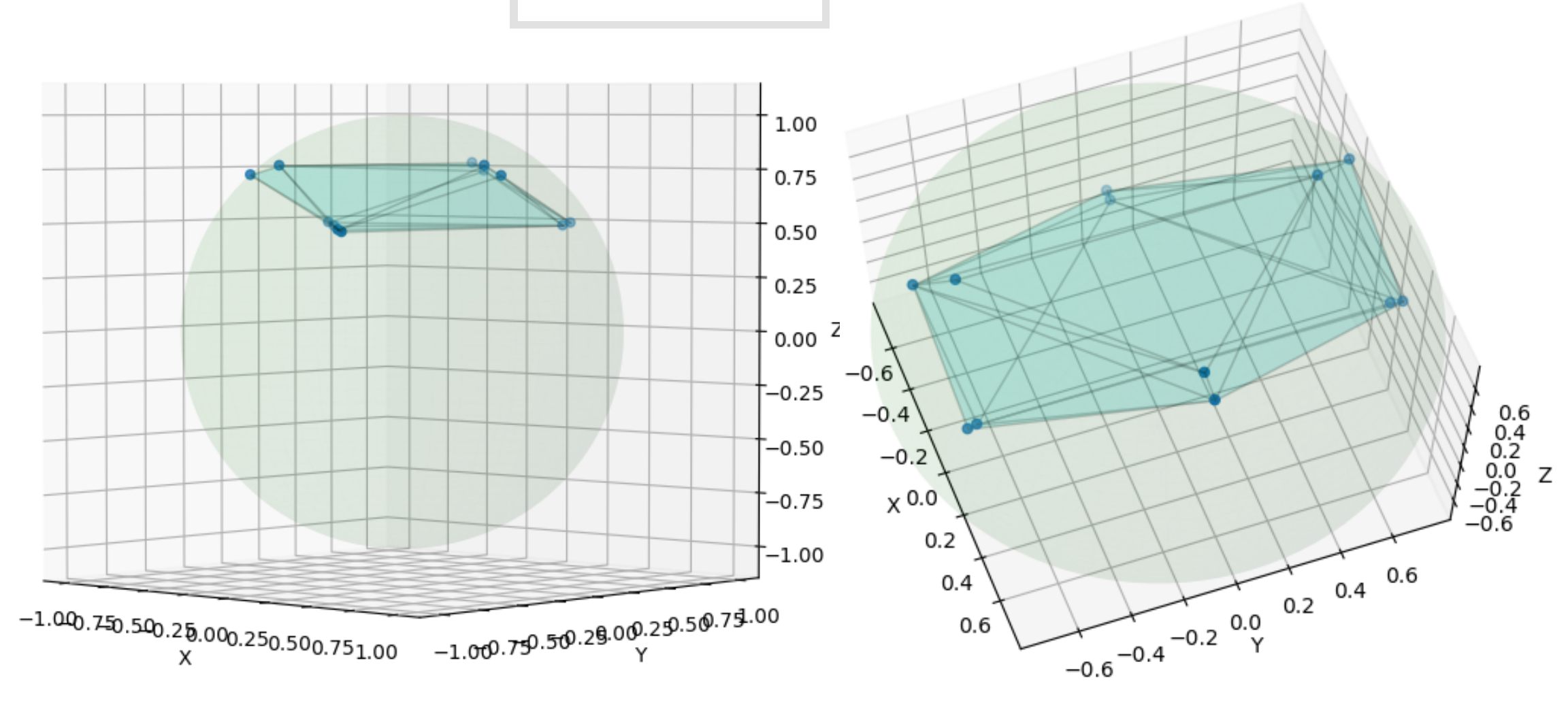}
    \caption{initial confidence polytope}
    \label{fig:demo1}
  \end{subfigure}
  \begin{subfigure}[t]{0.21\linewidth}
    \includegraphics[width=\linewidth]{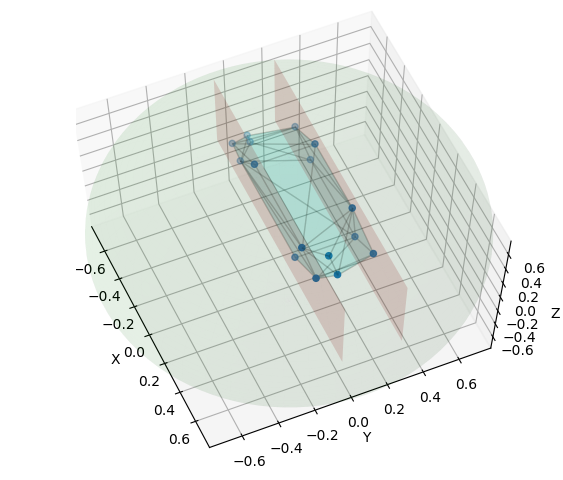}
    \caption{extra $Y$ measurements}
    \label{fig:demo2}
  \end{subfigure}
  \begin{subfigure}[t]{0.21\linewidth}
    \includegraphics[width=\linewidth]{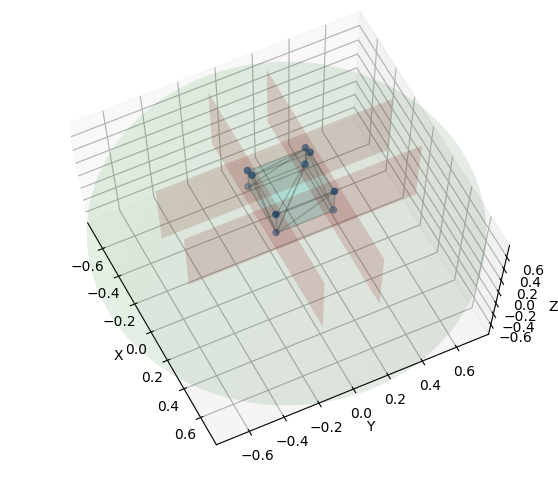}
    \caption{extra $X$ measurements}
    \label{fig:demo3}
  \end{subfigure}
\caption{{\bf Optimising the information content of measurements.}   The red dots in~(a) represent the elements of a skewed SIC POVM, which has more distinguishing power in the $Z$~direction than in the $X$ and $Y$~directions (cf. Appendix~\ref{appE} for a precise definition). (b) shows the confidence polytope that is obtained from $5000$ measurements defined by this POVM. (c) \& (d) depict the effect of $1000$ additional projective measurements in both the $X$ and the $Y$ direction. The red planes represent the new facets introduced by the extra measurements. \label{fig_LearningFrom Polytope}
  }
\end{figure*}

To illustrate this, we provide examples of simple QST scenarios. The first is QST on a single qubit, where the state space is three-dimensional, so that the confidence polytopes can  be depicted easily (Fig.~\ref{fig:polytope}). For higher-dimensional systems, the reporting of full polytopes is in general no longer sensible.  In these cases, one may instead extract certain 
\emph{figures of merit} from them, such as the fidelity of the unknown state with respect to a desired reference state.  The confidence interval of such a figure of merit is then determined by its maximum and minimum among all states of the polytope. In practice, these numbers may be estimated by sampling at random from the polytope. We demonstrate this for  QST on a noisy Bell state  (Table~\ref{tab:bell}) with simulated measurement data and on $s$-qubit GHZ states~\footnote{An $s$-qubit GHZ state is defined as $\ket{\psi}={(\ket{0}^{\otimes s}+\ket{1}^{\otimes s})/\sqrt{2}}$.} for $s=2,3,4$ (Table~\ref{tab:ibmq}) with  data from IBM's Q Experience~\cite{IBM} (cf.\ Appendix~\ref{appC}). For all our examples we chose a confidence level defined by $\eps=0.001$.

\begin{table}[h!]
  \begin{center}
    \begin{tabular}{l|c|c|c}
        & \text{Fidelity} & \text{Trace distance} & \text{Negativity}\\ 
      \hline
      \text{MLE state} & $>0.973$ & $<0.0902$ & \multirow{ 2}{*}{$(0.393, 0.459)$}\\
      \cline{1-3}
      \text{Perfect Bell state} & $(0.944, 0.980)$ & $(0.0546, 0.133)$ & 
    \end{tabular}
    \caption{{\bf QST of simulated noisy Bell state.} A confidence polytope with confidence level $0.999$ was generated for  data  from simulated SIC POVM measurements on $10^4$ copies of a noisy Bell state. The table shows the confidence intervals, which are derived from the confidence polytope, for various figures of merit, such as the fidelity to and the distance from particular reference states (MLE denotes the state obtained by Maximum Likelihood Estimation), or the negativity, which is a measure for entanglement \cite{Vidal}. }\label{tab:bell}
  \end{center}
\end{table}

\begin{table}[h!]
  \begin{center}
    \begin{tabular}{l|c|c|c|c}
       & \text{data size} $n$ & \text{Fidelity} & \text{Trace distance} & \text{Negativity}\\ 
      \hline
      \text{GHZ2 } &$9\times 1024$& $(0.903, 0.940)$ & $(0.131, 0.208)$ & $(0.318, 0.386)$\\ 
      \hline
      \text{GHZ3} & $27\times 1024$ & $(0.837, 0.869)$ & $(0.313, 0.371)$ & /\\
      \hline
      \text{GHZ4} & $81\times 1024$ & $(0.944, 0.980)$ & $(0.0546, 0.133)$ & /   
    \end{tabular}
    \caption{{\bf QST of GHZ states on IBM's Q Experience.} GHZ states of $2$, $3$, and $4$ qubits  were prepared on IBM's 5-qubit device ``ibmqx2'' and then measured with respect to the Pauli basis on each qubit. The sample size is given by the number of different measurement directions times the shot count (each measurement is repeated 1024 times). The third and fourth column show the deviation from perfect GHZ states. The confidence intervals are calculated for a $0.999$ confidence level. \label{tab:ibmq}}
  \end{center}
\end{table}

\section{Using Polytopes to Optimise Measurements} \label{sec_optimise}

The shape of the confidence polytopes  provides information about the distribution of the statistical errors. This, in turn, enables the choice of particular additional measurements to improve the precision of QST. We demonstrate this with a simple example of QST on a single qubit. (Its low dimensionality allows us to illustrate the idea by intuitive plots in the Bloch sphere picture | but a generalisation to higher-dimensional spaces is straightforward.) We start with a biased informationally complete POVM, which may be regarded as a skewed version of a SIC POVM (cf.\ Fig.~\ref{fig_LearningFrom Polytope}(a)). The polytope obtained after $5000$ measurements is much more extended in the $X$ and the $Y$ direction than in the $Z$ direction (Fig.~\ref{fig_LearningFrom Polytope}(b)). Therefore, $1000$ extra measurements along each the $X$ and the $Y$ direction  help to ``refine'' the polytope, yielding a smaller confidence region (Figs.~\ref{fig_LearningFrom Polytope}(c) and~(d)).  

In higher dimensions, extracting the relevant geometrical information can be computationally expensive. One may however simplify this task by considering a \emph{bounding box} in the representation space $\smash{\mathbb{R}^{d^2-1}}$, defined as the minimum enclosing hyper-rectangle with faces perpendicular to the axes given by the basis $\{\lambda_i\}$. Ideally, this basis should be chosen such that it contains experimentally accessible observables (e.g., tensor powers of Pauli matrices).  Since the orientation of the bounding box is fixed by the basis, the corners of the box can be determined via simple linear programs. If a particular edge of the bounding box is long, it implies that the confidence polytope is more extended in that direction, and further measurements along the corresponding axis would be effective in reducing its size.

\section{Comparison to Bayesian Credibility Regions} \label{sec_comparison}

Confidence regions have the advantage over  Bayesian credible regions that they do not rely on any prior knowledge. Conversely, credibility regions are generally smaller than confidence regions, thus giving tighter state estimates~\cite{Jaynes}. Clearly, if the prior is already highly peaked around the actual (unknown) state of the system, the credibility regions obtained by QST can be arbitrarily small. However, numerical results indicate that, in the case of relatively flat priors, the resulting credibility regions turn out to be comparable in size to the confidence polytopes introduced here. 

Specifically, we take priors defined by the Hilbert-Schmidt measure $d \rho$~\footnote{Given the Haar measure, $d \phi$, over the purification space ${\cH \otimes \cK}$, the Hilbert-Schmidt measure over $\cS(\cH), d \rho$, is induced by tracing out $\cK$.}. A region $\Gamma(B^{\mathbf{n}})$ of the state space has credibility ${1-\eps_b}$ with respect to this prior if the condition 
\begin{align} \label{baye}
\int_{\Gamma(B^{\mathbf{n}})}\mu_{B^{\mathbf{n}}}(\rho) d\rho\geq 1-\eps_b \ ,
\end{align}
with $\mu_{B^{\mathbf{n}}}(\rho)d\rho=\smash{\frac{\tr[B^{\mathbf{n}}\rho^{\otimes n}]d\rho}{\int\tr[B^{\mathbf{n}}\rho^{\otimes n}]d\rho}}$, holds. For our comparison, we take $\Gamma(B^{\mathbf{n}})$ to be a $({1-\eps})$-confidence polytope as in Theorem~\ref{main} and determine its credibility level $\eps_b$ by~\eqref{baye}. We then plot the ratio $\eps/\eps_b$ for randomly chosen states.  As shown in Fig.~\ref{fig:scaling}, the numerics indicate that this ratio does not scale with the dimension of the measured quantum system nor with the data size~\footnote{The computation was carried out with the \emph{Tomographer} software by Faist~\cite{Faist}.}. Confidence polytopes therefore provide rather tight estimates for the unknown state. In particular, they outperform the earlier construction proposed by Christandl \& Renner~\cite{ChristandlRenner}. In the latter, $({1-\eps})$-confidence regions are constructed from particular $({1-\eps_b})$-credibility regions, but $\eps$ is larger than $\eps_b$ by a factor polynomial in the dimension of the measured system (cf.\ Appendix~\ref{appD}). 

\begin{figure}
  \centering
  \begin{subfigure}[b]{0.49\linewidth}
    \includegraphics[width=\linewidth]{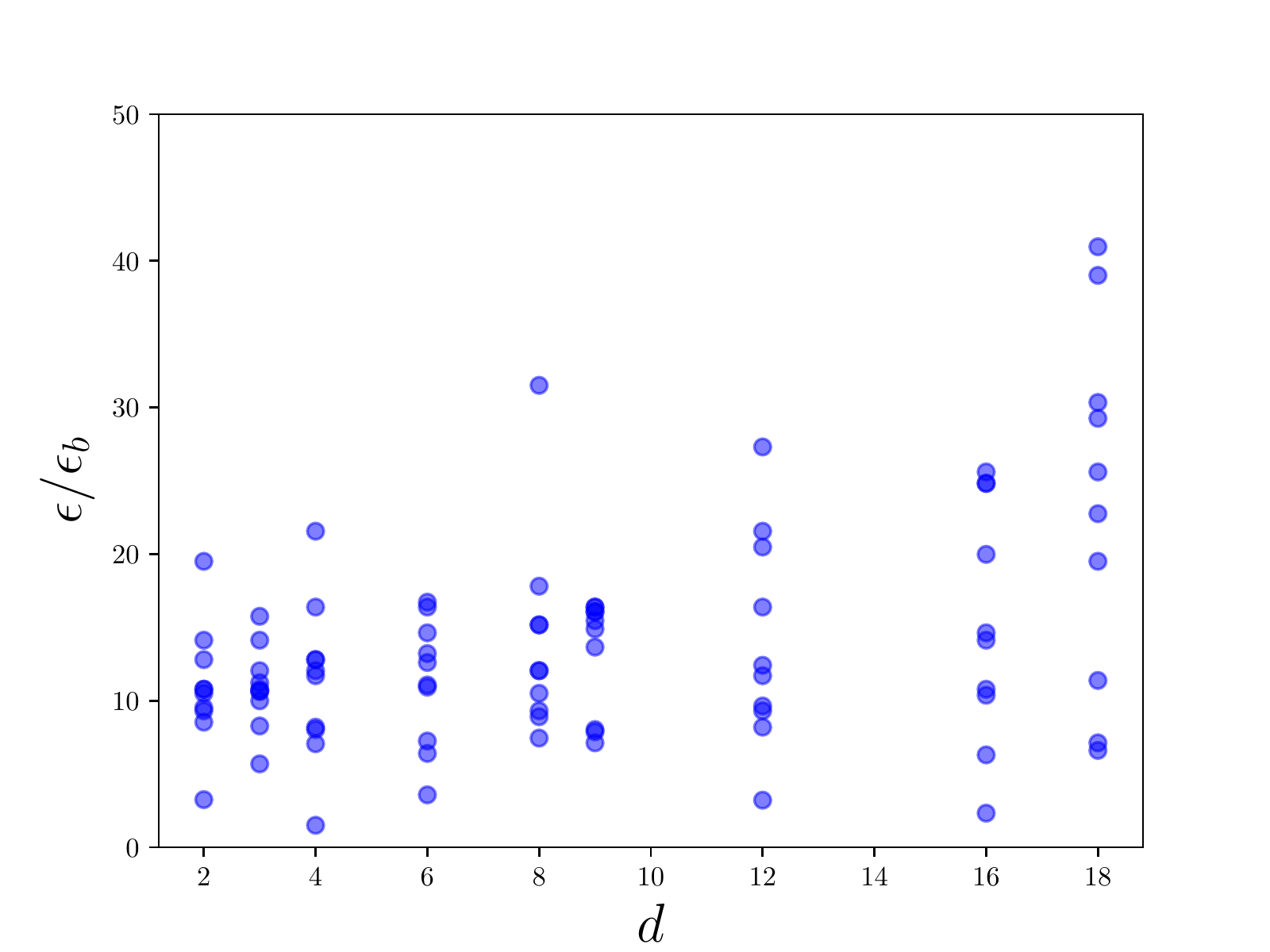}
    \caption{Dependence on dimension}
    \label{fig:scalingdim}
  \end{subfigure}
  \begin{subfigure}[b]{0.49\linewidth}
    \includegraphics[width=\linewidth]{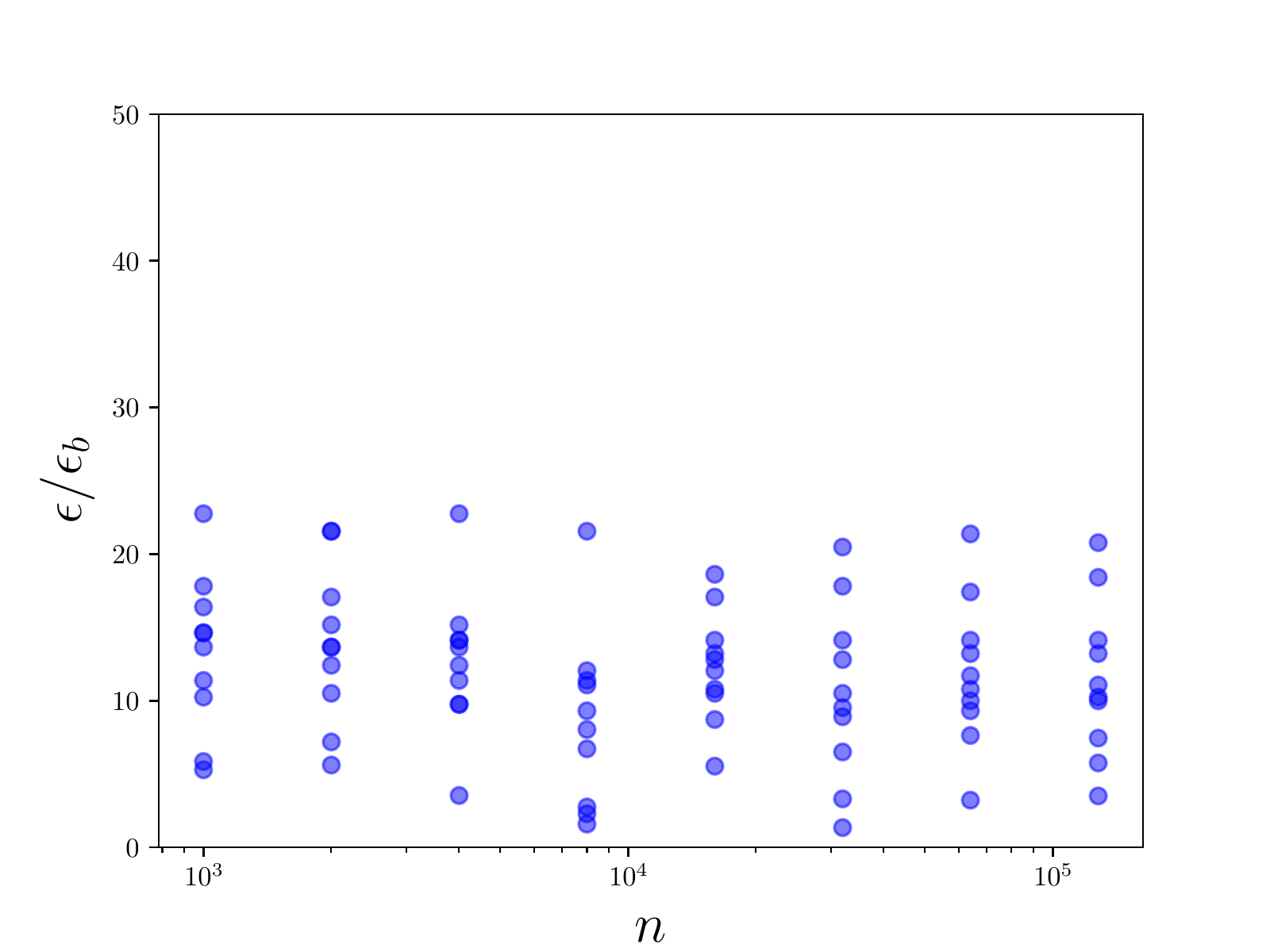}
    \caption{Dependence on data size}
    \label{fig:scalingn}
  \end{subfigure}
\caption{{\bf Confidence vs.\ credibility regions.}  The plots show the ratio $\eps/\eps_b$ between the confidence and credibility levels of a polytope constructed according to the prescription of Theorem~\ref{main}, interpreted as a confidence region and as a credibility region, respectively.  Each dot was obtained by QST on $n$ copies of a state chosen at random from a $d$-dimensional state space. Although there are fluctuations due to the different choices of states and measurement outcomes, no scaling in the data size~$n$ or the dimension~$d$ is observed. 
  }
 \label{fig:scaling}
\end{figure}

\section{Discussion} \label{sec_discussion}

Our work is a first attempt to generalise  methods from classical confidence interval estimation to QST.  The resulting construction of confidence polytopes provides various improvements over previously known methods. In particular, confidence polytopes are comparable in size to credibility regions, and they are efficiently computable. 

Apart from the Clopper-Pearson confidence intervals that we used for our construction, there exist several other methods to determine confidence intervals in classical statistics, many of which rely however on approximations~\cite{AgrestiCoull,BlythStill,Wilson,Thulin,Sakakibara,BrownCaiDasGupta01}.  Some of these methods yield confidence intervals that are smaller than Clopper-Pearson intervals and thus seem to have more prediction power~\cite{Thulin,BrownCaiDasGupta01}.  Conversely, the Clopper-Pearson confidence intervals are a safe choice, in the sense that they never result in an overestimation of the confidence level. Furthermore, for sample sizes~$n$ of the order $10^5$, which is common in QST, the actual coverage probability of the Clopper-Pearson intervals is very close to the claimed confidence level~\cite{Thulin}.  Having said this, an interesting alternative may be \emph{Agresti-Coull confidence intervals}~\cite{AgrestiCoull,BrownCaiDasGupta01,BrownCaiDasGupta02}. These are also efficiently computable and their coverage probability fluctuates with small amplitude around the claimed confidence level.

Confidence polytopes as proposed here may also  be combined with methods for  dimension-scalable QST. These are based on additional assumptions about the unknown state, e.g., that it has bounded rank~\cite{Baldwin}, that it is permutation-invariant~\cite{Toth}, or that it has a matrix product state (MPS) structure~\cite{Cramer}. These assumptions generally restrict the relevant state space. Accordingly, it is sufficient to construct confidence polytopes within this restricted space.  

While we argued that the result of QST should be confidence regions rather than point estimates, one may still ask whether sensible state estimates can be obtained. However, the bias of such estimates, or ``˜quantum gerrymandering'', as Carlton Caves put it~\cite{Caves}, appears to be a general problem, especially when  the unknown state is close to the positivity boundary of the state space. For example, the state estimate using MLE is biased towards a rank-$2$ density operator when the unknown state is a pure qutrit~\cite{Caves}. In our case, when the polytope intersects with the boundary,  the part of the polytope outside the state space is truncated, and thus a state estimate drawn from our confidence region would be biased. Whether and how a state estimate without such bias can be obtained is still an open question, whose answer seems to require a better understanding of the boundary geometry of the state space~\cite{Scholten}. 

As shown above, rather than reporting the full confidence polytope as the outcome of a QST experiment, it is often sensible to characterise it with one (or a few) parameters. There is of course a lot of freedom in how to do so. One possibility is to report the centre of the smallest enclosing ball with respect to a given distance measure (also known as the Chebyshev ball~\cite{Boyd}), and take its radius as the error bar. Alternatively, one could treat the state obtained from any point estimation scheme, such as MLE or \emph{Constrained Least Square}, as a reference and report its maximum distance to any of the vertices of the polytope as the error bar~\footnote{Determining this distance exactly may be difficult for high-dimensional systems, because the vertex enumeration problem is in general NP-hard~\cite{Khachiyan}. There also exists no efficient algorithm to compute the size of the polytope whilst taking into account the positivity boundary~\cite{Suess}.  An alternative, however, is to calculate an approximate solution by sampling within the polytope, as we did it for the examples in this paper.}. In this sense, our methods, rather than replacing current state estimation schemes,  endow them with error bars that characterise the statistical (un)certainty of the estimates.  We look forward to test these methods in future collaborations with experimental groups.

\section{Acknowledgements}

We thank Philippe Faist for providing the \emph{Tomographer} software, and Rotem Arnon-Friedman, Kenny Choo, Johannes Heinsoo, Raban Iten, Joeseph Renes, Ernest Tan, and Chi Zhang for useful discussions. This work received support from the Swiss National Science Foundation via the National Center for Competence in Research ``QSIT''.

 \bibliographystyle{apsrev}

 \providecommand \doibase [0]{http://dx.doi.org/}%

\onecolumngrid

\appendix

\newpage
\section*{{\Large Appendices}}
\bigskip

\renewcommand{\thesubsection}{\Alph{subsection}}

\subsection{Proofs of the Main Claims}\label{appA}

\setcounter{lemma}{0}
\setcounter{theorem}{0}
\setcounter{corollary}{0}
\setcounter{equation}{0}


\noindent {\it Proof of Theorem~\ref{main}.} 
The proof is subdivided into two parts. In the first we  reduce the conditions
\begin{align}\label{ineq1}
\rho \in \Gamma_i(n_i) \;\;\; \iff \;\;\;&  \tr (E_i \rho) \leq \frac{n_i}{n} + \delta(n_i, \eps_i)
\end{align}
to conditions that characterise the Clopper-Pearson construction, i.e., 
$$\Gamma'(B^{\mathbf{n}}):=\bigcap_i\Gamma'_i(n_i) \ ,$$ 
where the regions $\Gamma_i'(n_i)$ may be defined via their complements as 
\begin{align}\label{ineq2}
\rho \in \Gamma'^c_i(n_i)\;\;\; \iff \;\;\;&\sum_{j=0}^{n_i}\binom{n}{j}(\tr E_i\rho)^j(1-\tr E_i\rho)^{n-j} \leq \eps_i \ .
\end{align}
In the second, we show that the latter gives proper confidence regions in quantum state space. 

Note that we can treat the sample mean $\frac{n_i}{n}$ as a deviation $\delta(\rho)$ from the expectation $\tr E_i\rho$, i.e., $\tr E_i\rho=\frac{n_i}{n}+\delta(\rho)$, and hence rewrite the condition in~(\ref{ineq2}) as $$\sum_{j=0}^{n_i}\binom{n}{j}(\frac{n_i}{n}+\delta(\rho))^j(1-\frac{n_i}{n}-\delta(\rho))^{n-j} \leq \eps_i \ . $$
The Chernoff-Hoeffding Theorem~ \cite{Hoeffding,Chernoff} asserts that, for any $\delta > 0$,
$$\sum_{j=0}^{n_i}\binom{n}{j}(\frac{n_i}{n}+\delta)^j(1-\frac{n_i}{n}-\delta)^{n-j} \leq e^{-nD(\frac{n_i}{n}||\frac{n_i}{n}+\delta)}  $$  
where $D(x||y)=x\log(\frac{x}{y})+(1-x)\log(\frac{1-x}{1-y})$ is the Kullback-Leibler divergence between Bernoulli distributed random variables. Hence, taking $\delta = \delta(n_i, \eps_i)$ to be the positive real that solves $D(\frac{n_i}{n}||\frac{n_i}{n}+\delta)=-\frac{1}{n}\mathrm{log}(\eps_i)$, we have
 $$\sum_{j=0}^{n_i}\binom{n}{j}(\frac{n_i}{n}+\delta)^j(1-\frac{n_i}{n}-\delta)^{n-j} \leq \eps_i \ . $$
Suppose now that $\rho \notin \Gamma_i(n_i)$. According to~(\ref{ineq1}), we then have $\delta(\rho) > \delta$.  Because the function $f(p)=p^j(1-p)^{n-j}$ has negative gradient for $p>j/n$, we get
 $$\sum_{j=0}^{n_i}\binom{n}{j}(\frac{n_i}{n}+\delta(\rho))^j(1-\frac{n_i}{n}-\delta(\rho))^{n-j} \leq\sum_{j=0}^{n_i}\binom{n}{j}(\frac{n_i}{n}+\delta)^j(1-\frac{n_i}{n}-\delta)^{n-j} \leq \eps_i \ . $$
Hence, by virtue of~\eqref{ineq2}, $\rho \in \Gamma'^c_i(n_i)$, from which we can conclude $$\Gamma^c_i(n_i) \subseteq\Gamma'^c_i(n_i) \implies \Gamma'_i(n_i) \subseteq\Gamma_i(n_i) \implies \Gamma'(B^{\mathbf{n}}) \subseteq\Gamma(B^{\mathbf{n}}) \ .$$ 

We now start with the second part of the proof, in which we show that the Clopper-Pearson condition~(\ref{ineq2}) yields a confidence region in the QST scenario. More specifically, we show that the probability that $\Gamma'(B^{\mathbf{n}})$ does not contain the unknown state $\rho$ is upper bounded by~$\eps$. Denoting by $\chi$  the indicator function, we have
\begin{align*}
\mathrm{Pr}[\rho \in \Gamma'^c(B^{\mathbf{n}})]&=\sum_{\mathbf{n}} \tr B^{\mathbf{n}}\rho^{\otimes n} \chi(\rho \in \bigcup_i \Gamma'^c_i(n_i))\\
&\leq \sum_{i}\sum_{\mathbf{n}} \tr B^{\mathbf{n}}\rho^{\otimes n} \chi(\rho \in \Gamma'^c_i(n_i)) \\
&=\sum_{i} \sum_{\mathbf{n}} n!\prod_j \frac{1}{n_j!}\tr(E_j^{\otimes n_j}\rho^{\otimes n})\chi(\rho \in \Gamma'^c_i(n_i)) \\
&=\sum_{i} \sum_{n_i=0}^{n} \binom{n}{n_i}\tr(E_i^{\otimes n_i}(\id-E_i)^{\otimes n-n_i}\rho^{\otimes n})\chi(\rho \in \Gamma'^c_i(n_i)) \ ,
\end{align*}
where, for the last equality, we used that $\chi(\rho \in \Gamma'^c_i(n_i))$ is independent of $n_j$ for any $j \neq i$, so that the sum over these values $n_j$ can be easily evaluated. 

It is obvious from~(\ref{ineq2}) that $\Gamma'^c_i$ shrinks monotonically for increasing $n_i$, i.e., $n_i < n'_i \implies \Gamma'^c_i(n'_i) \subset \Gamma'^c_i(n_i)$. We can therefore define a maximum value $m_i$ such that $\rho \in \Gamma'^{c}_i(m_i)$ but $\rho \notin \Gamma'^{c}_i(n_i)$ for any $n_i > m_i$. (It may also be that there is no value $n_i$ for which $\rho \in \Gamma'^c_i(n_i)$, in which case we set $m_i < 0$ in the expressions below.) The above sum may then be rewritten as
\begin{align*}
\mathrm{Pr}[\rho \in \Gamma'^c(B^{\mathbf{n}})]&\leq\sum_{i} \sum_{n_i=0}^{m_i} \binom{n}{n_i}\tr(E_i^{\otimes n_i}(\id-E_i)^{\otimes n-n_i}\rho^{\otimes n}) \ .
\end{align*}
Finally, because $\rho \in \Gamma'^{c}_i(m_i)$, it follows from (\ref{ineq2}) that
\begin{align*}
\mathrm{Pr}[\rho \in \Gamma'^c(B^{\mathbf{n}})]&\leq\sum_i \sum_{n_i=0}^{m_i}\binom{n}{n_i}(\tr E_i\rho)^{n_i}(1-\tr E_i\rho)^{n-n_i}\\
&\leq \sum_i \eps_i=\eps \ .
\end{align*}

Since $\Gamma'(B^{\mathbf{n}}) \subseteq \Gamma(B^{\mathbf{n}})$, as shown in the first part of the proof, we conclude that $\mathrm{Pr}[\rho \in \Gamma^c(B^{\mathbf{n}})]\leq \eps$.

\qed


\noindent { \it Proof of Corollary~\ref{mainc}.}
Plugging the parametrisations into the Born rule, the conditions of Theorem~\ref{main} can be written as $$\frac{1}{m_i}(1+(d-1)\mathbf{r} \cdot \boldsymbol{\eta}_i)=\tr E_i\rho \leq \frac{n_i}{n}+\delta (n_i,\eps_i) \ . $$
where we used the orthogonality conditions $\tr \lambda_i \lambda_j =2\delta_{i j}$ and the fact that the Pauli operators have trace zero.
Hence, for each~$i$, we obtain an Euclidean half-space with the normal corresponds to the vector $\boldsymbol{\eta}_i$. The intersection of them corresponds to the confidence region in $\mathbb{R}^{d^2-1}$. 

\qed

\begin{remark}
If the POVM is informationally complete, there are at least $d^2$ linearly independent vectors $\boldsymbol{\eta}_i$ and $d^2$ corresponding half-spaces, which is enough to form a  convex polytope in $\mathbb{R}^{d^2-1}$.
\end{remark}

When data from separate measurements with respect to different POVMs (e.g., projective measurements along various axes) on the same identically prepared system are available,  the confidence polytopes of the individual measurements may be combined into one single polytope, as asserted by the following lemma.

\begin{lemma}
Consider $M$ measurements and suppose that measurement $j \in \{1, \ldots, M\}$ has a confidence region $\Gamma^{(j)}$ with confidence level $1-\eps^{(j)}$. Then the intersection of all of them is a confidence region $\Gamma$ with confidence level  $1-\sum_{j=1}^M \eps^{(j)}$.
\end{lemma}

\begin{proof}
The proof is elementary, 
\begin{align*}
\mathrm{Pr}[\rho \in \Gamma^c]
=\mathrm{Pr}[\rho \in \bigcup_j (\Gamma^{(j)})^c] 
\leq \sum_j \mathrm{Pr}[\rho \in (\Gamma^{(j)})^c] 
\leq \sum_j \eps^{(j)} \ .
\end{align*}

\end{proof}

We point out, however, that in practical situations it may be better to choose the different measurements  stochastically to circumvent a systematic drift. In this case, the measurements should be modelled as part of one single POVM with appropriately weighted elements.

\subsection{Confidence Polytopes with Higher Order Facets}\label{appB}

\begin{figure}[t]
  \centering
  \begin{subfigure}[b]{0.49\linewidth}
    \includegraphics[width=\linewidth]{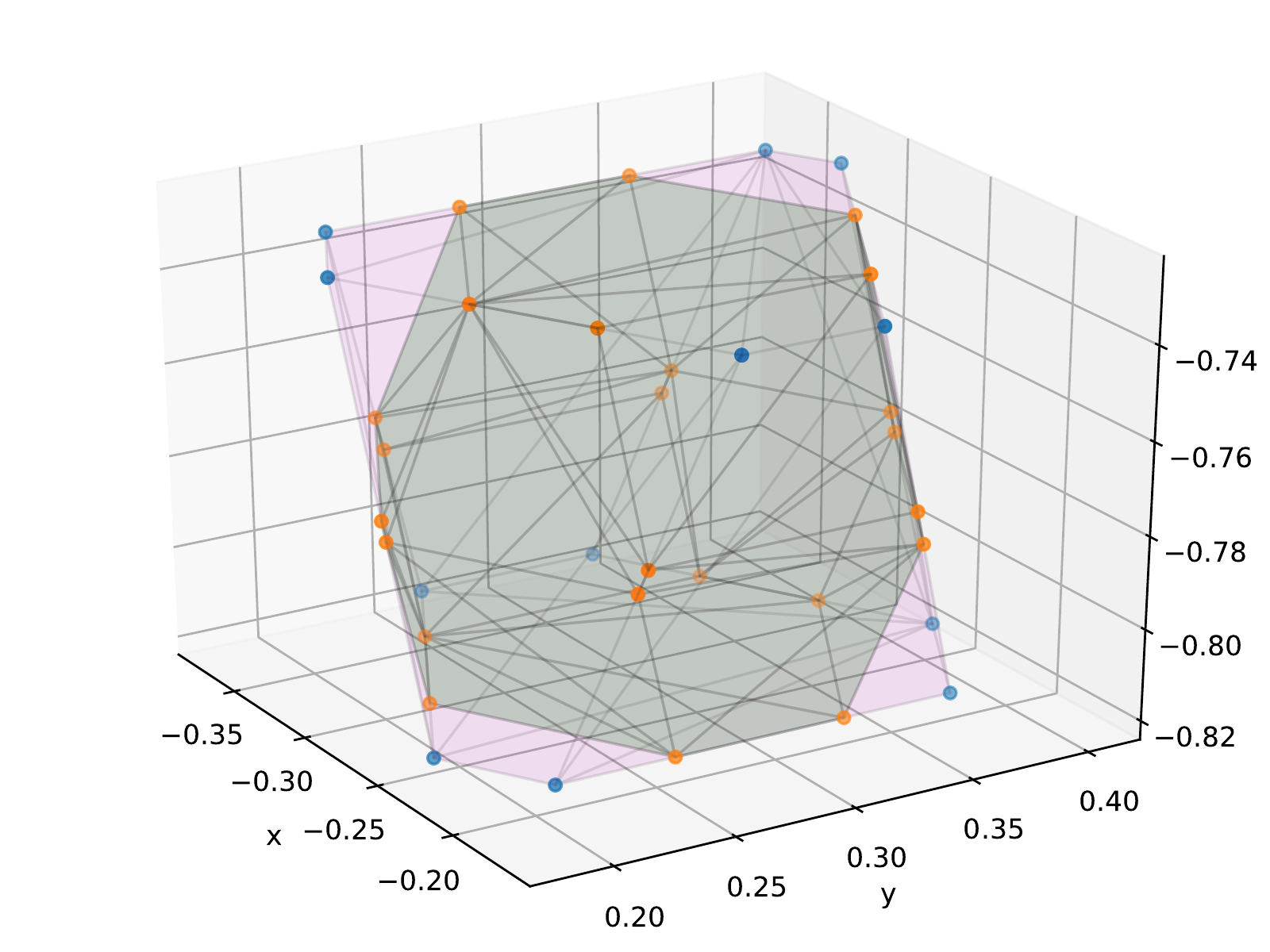}
    \caption{SIC POVM}
    \label{fig:morefacets1}
  \end{subfigure}
  \begin{subfigure}[b]{0.49\linewidth}
    \includegraphics[width=\linewidth]{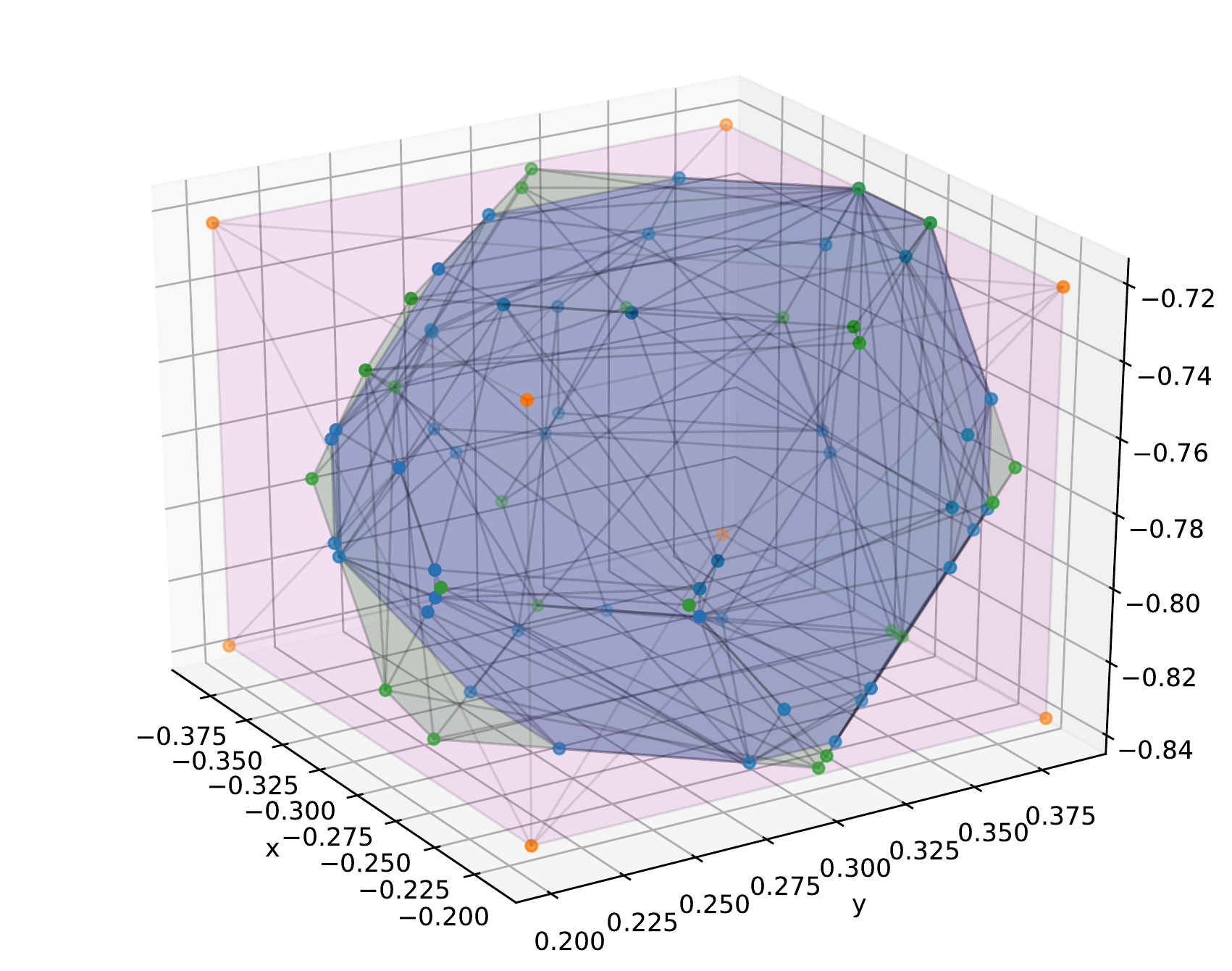}
    \caption{MUB POVM}
    \label{fig:morefacets2}
  \end{subfigure}
  \caption{{\bf Higher order facets.} The plots illustrate the effect of including more facets in the specification of the confidence polytope. These additional facets are obtained by grouping  POVM elements to form new ones. (a)~shows the result of QST on a qubit with a SIC POVM, which has $4$ elements. The grey-green confidence polytope is a refinement of the pink one, obtained by adding facets. Each of the additional facets corresponds to a grouping of the POVM elements according to the scheme $(2,2)$, i.e., each of the groups consists of $2$ elements. (b)~shows the outcome of QST with measurements defined by a MUB POVM composed of $6$ elements. The green and blue confidence polytopes are refinements of the pink polytope obtained by grouping these elements according to the scheme $(2,4)$ and $(3,3)$, respectively.}
  \label{fig:refine}
\end{figure}

The confidence polytope we constructed has one facet for each of the POVM elements  $E_i$ (cf.\ Theorem~\ref{main} and Corollary~\ref{mainc}). However, we can obtain more facets by considering groupings of the POVM elements. To illustrate this, consider QST on a qubit with measurements defined by a SIC POVM, which contains $4$ elements. By grouping them as  $\{E_1+E_2, E_3+E_4\}$, $\{E_1+E_3, E_1+E_4\}$, and $\{E_1+E_4, E_2+E_3\}$, we can obtain $6$ additional facets, as shown in Fig.~\ref{fig:refine}(a).  The effect of introducing these additional facts is even more striking in the case of a MUB POVM, as shown in Fig.~\ref{fig:refine}(b). The only prize we have to pay for these additional facets is that the confidence level $\eps$ must be split into a larger number of components~$\eps_i$. But since these enter the definition of the confidence region only logarithmically, this effect is negligible as long as the  POVM has only few elements. However, for a general POVM with $k$ elements, the number of  facets that can in principle be constructed by such grouping can be as large as $2^k-2$, where $k\geq d^2$ for an informationally complete POVM. For higher-dimensional systems, determining a grouping that yields a polytope of minimal size is thus a non-trivial optimisation problem, which we leave for future work.  

\subsection{Quantum State Tomography with IBM's Q Experience } \label{appC}

\begin{figure}[t]
  \centering
  \begin{subfigure}[t]{0.4\linewidth}
    \includegraphics[width=\linewidth]{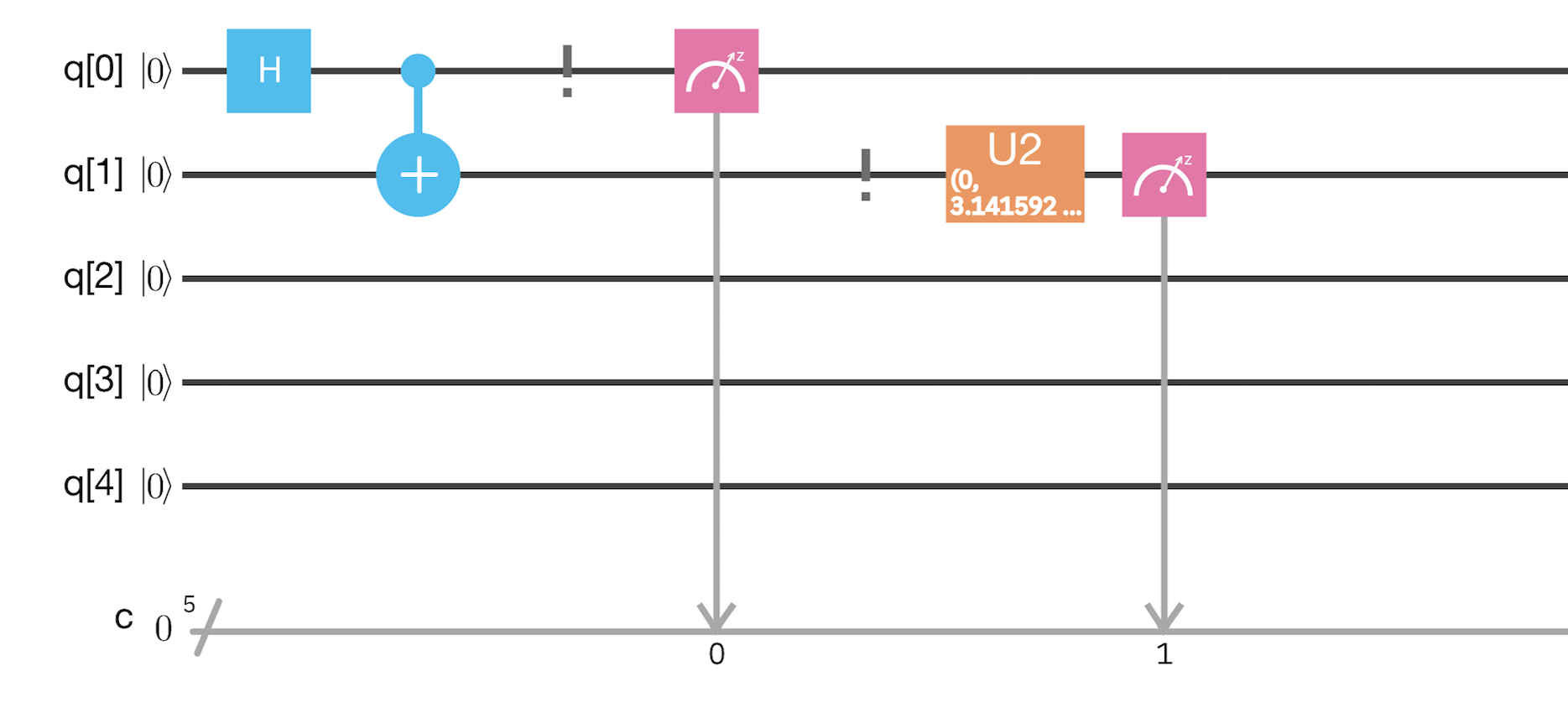}
    \caption{GHZ2}
  \end{subfigure}
  \begin{subfigure}[t]{0.49\linewidth}
    \includegraphics[width=\linewidth]{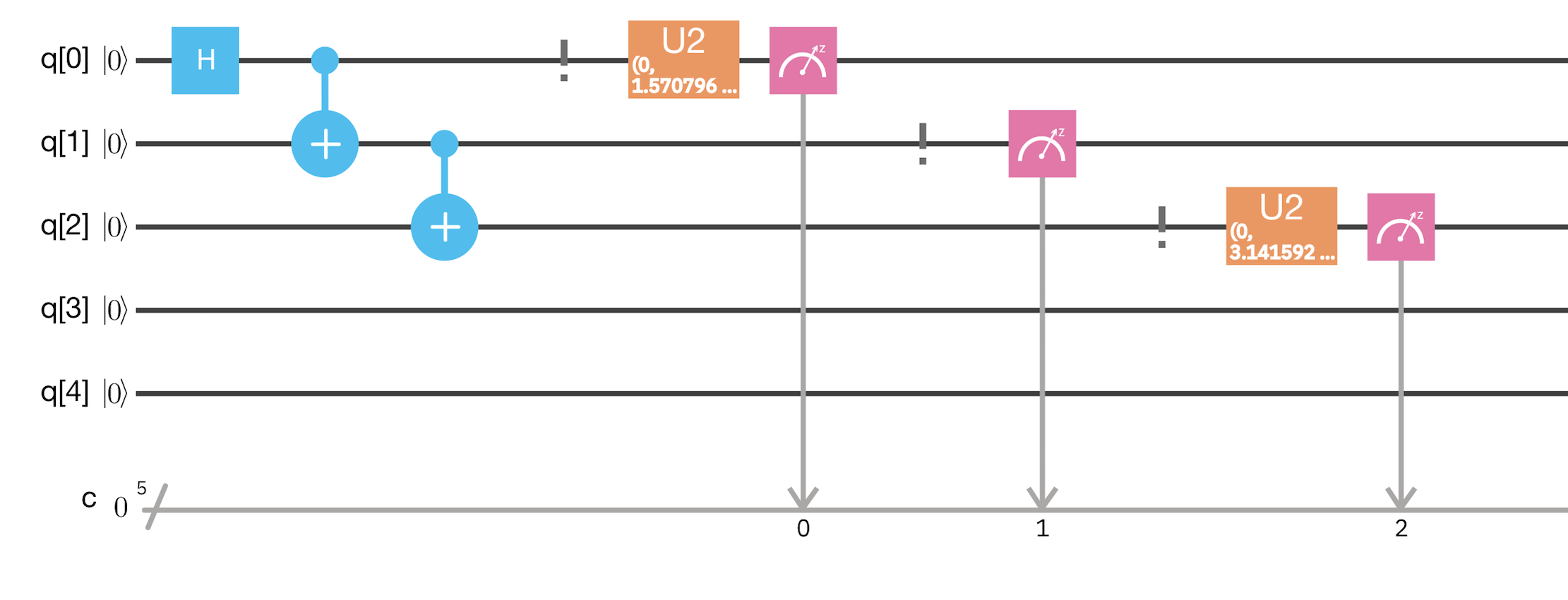}
    \caption{GHZ3}
  \end{subfigure}
    \begin{subfigure}[b]{0.8\linewidth}
    \includegraphics[width=\linewidth]{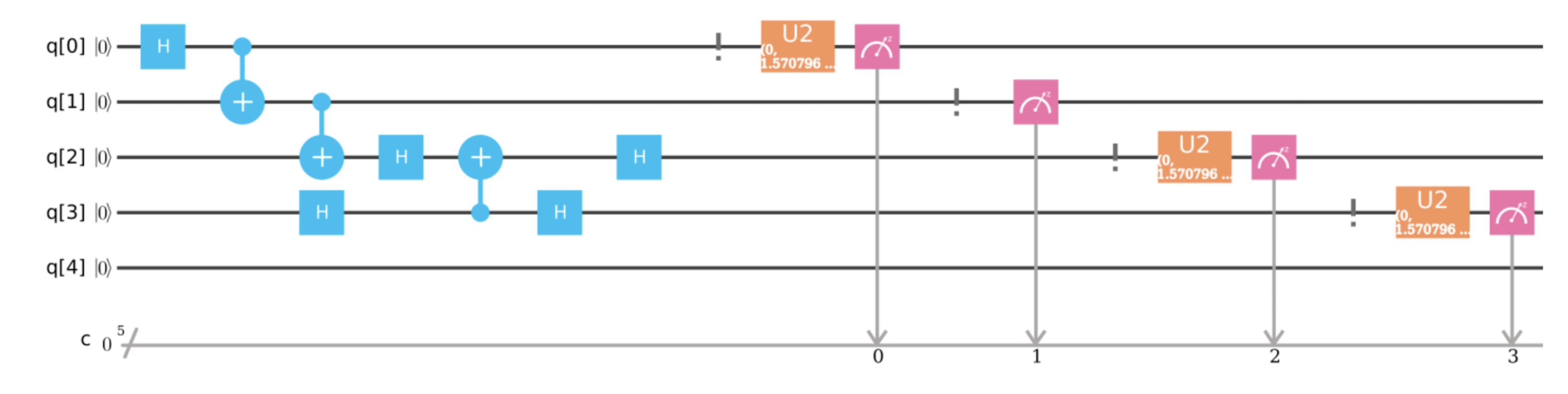}
    \caption{GHZ4}
  \end{subfigure}
  \caption{{\bf Tomography on IBM's Q Experience.} The figure shows the circuits used for our QST experiment with GHZ states  of different sizes on ``ibmqx2''.  }
  \label{fig:circuit}
\end{figure}

\begin{table}
  \begin{center}
    \begin{tabular}{l|c|c|c|c|c}
       & \text{Q0} & \text{Q1} & \text{Q2} & \text{Q3}& \text{Q4}\\ 
      \hline
      \text{Frequency(GHz)} &$5.28$& $5.21$ & $5.02$ & $5.28$& $5.07$\\ 
      \hline
      \text{T1($\mu s$)} & $59.40$ & $67.80$ & $68.90$ & $48.90$& $66.00$ \\ 
      \hline
      \text{T2($\mu s$)} & $41.50$ & $55.30$ & $67.10$ & $69.80$& $44.20$  \\ 
      \hline
      \text{Gate error ($10^{-3}$)} & $1.97$ & $1.29$ & $1.97$ &$1.63$ &$0.94$  \\ 
      \hline
      \text{Readout error ($10^{-2}$)} & $4.50$ & $3.60$ & $2.00$ &$1.60$ &$2.50$ \\ 
      \hline
      \text{MultiQubit gate error ($10^{-2}$)} &CX0-1:3.46; CX0-2:4.07 & CX1-2:3.26 &/ & CX3-2:2.76; CX3-4:2.23 & CX4-2: 2.66  \\ 
    \end{tabular}
    \caption{{\bf Calibration data of ``ibmqx2'':} The table shows the calibration information at the time of the  QST experiments, when the lastest calibration time was 2018-01-29 05:41:13.}\label{tab:calibration}
  \end{center}
\end{table}

For the results shown in Table~\ref{tab:ibmq}, GHZ states of~$2$, $3$, and~$4$ qubits were prepared on the $5$-qubit device ``ibmqx2''™, and measured with respect to single-qubit Pauli operators. Data was collected for $1024$ shots in each direction.  The connectivity map for the CNOT on ``ibmqx2''™ is  $\{0: [1, 2], 1: [2], 3: [2, 4], 4: [2]\}$, where $a: [b]$ means a CNOT with qubit a as control and b as target can be implemented. The corresponding circuits that implement the GHZ states  GHZ2, GHZ3, GHZ4 are shown in Fig~\ref{fig:circuit}. We refer to  Table~\ref{tab:calibration} for the calibration information.  Further details can be found at:~\textit{ https://github.com/QISKit/ibmqx-backend-information/blob/master/backends/ibmqx2/README.md}.

\subsection{Comparison to Christandl \& Renner Confidence Regions} \label{appD}

\begin{figure}[t]
  \centering
  \begin{subfigure}[t]{0.4\linewidth}
    \includegraphics[width=\linewidth]{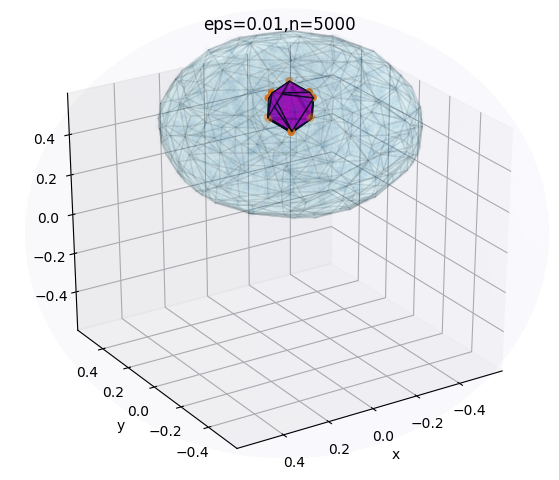}
  \end{subfigure}
  \caption{{\bf Comparison to the Christandl \& Renner confidence region:} The purple region is our confidence polytope and the light blue one is the Christandl \& Renner confidence region. The confidence level was set to 0.99 and the data obtained from 5000 repetitions of (simulated) SIC POVM measurements on an identically prepared qubit. 
  }
  \label{fig:sizes}
\end{figure}

The numerical results shown in Fig.~\ref{fig:scaling} indicate that the error ratio $\eps/\eps_b$ of our confidence polytopes is of the order $10$, whereas the Christandl \& Renner confidence regions, by construction, have an error ratio of the order $\mathrm{poly}(n)$, where $n$ is the data size. Hence, compared to the latter,  confidence polytopes  are substantially smaller. Fig.~\ref{fig:sizes} illustrates the relative sizes for the case of QST on a single qubit.

\subsection{Skewed SIC POVM}\label{appE}

The skewed SIC POVM used for the example shown in Fig.~\ref{fig_LearningFrom Polytope} is defined as
$E_i=\frac{1}{2}(\id +\boldsymbol{\eta}_i \cdot \boldsymbol{\sigma})$
with
$$\boldsymbol{\eta}_0=(0,0,1), \boldsymbol{\eta}_1=(3/10,0,-1/3), \boldsymbol{\eta}_2=(-3/20, 3/20, -1/3), \boldsymbol{\eta}_3=(-3/20, -3/20, -1/3) \ .$$

\bibliographystyle{apsrev}

 \providecommand \doibase [0]{http://dx.doi.org/}%

\end{document}